\documentclass[USenglish,11pt,a4paper]{article}
\usepackage[a4paper,margin=1in]{geometry}
\usepackage[utf8]{inputenc}
\usepackage{amsmath,amsfonts,amssymb}
\usepackage{nicefrac}

\usepackage{amsthm} %
\usepackage{hyperref} %
\usepackage[capitalise,nameinlink,noabbrev]{cleveref}

\usepackage{graphicx}
\usepackage{subcaption}
\graphicspath{{img/}}

\usepackage{mathtools}

\usepackage{refcount}

\DeclareMathOperator{\dist}{dist}

\newcommand{\modplus}{\boxplus}
\newcommand{\suc}[1]{#1 \modplus 1}
\newcommand{\prob}[1]{\textsc{#1}}
\newcommand{\shp}[1]{\pi\left(#1\right)}

\newtheorem{theorem}{Theorem}
\newtheorem{lemma}[theorem]{Lemma}
\newtheorem{definition}[theorem]{Definition}
\newtheorem{claim}[theorem]{Claim}

\makeatletter
\newenvironment{claimproof}[1][\proofname]{
  
  \pushQED{\qed}%
  \normalfont \topsep6\p@\@plus6\p@\relax
  \trivlist
  \item[\hskip\labelsep\itshape
    #1\@addpunct{.}]\ignorespaces
}{%
  \popQED\endtrivlist%
}
\makeatother

\usepackage{enumerate}

\begin{document}

\title{W[1]-Hardness of the $k$-Center Problem\\Parameterized by the Skeleton Dimension}
\author{Johannes Blum\\University of Konstanz, Germany\\\texttt{johannes.blum@uni-konstanz.de}}
\date{}
\maketitle

\begin{abstract}
In the \prob{$k$-Center} problem, we are given a graph $G=(V,E)$ with positive edge weights and an integer $k$ and the goal is to select $k$ center vertices $C \subseteq V$ such that the maximum distance from any vertex to the closest center vertex is minimized. On general graphs, the problem is NP-hard and cannot be approximated within a factor less than $2$.

Typical applications of the \prob{$k$-Center} problem can be found in logistics or urban planning and hence, it is natural to study the problem on transportation networks. Such networks are often characterized as graphs that are (almost) planar or have low doubling dimension, highway dimension or skeleton dimension. It was shown by Feldmann and Marx that \prob{$k$-Center} is W[1]-hard on planar graphs of constant doubling dimension when parameterized by the number of centers $k$, the highway dimension $hd$ and the pathwidth $pw$~\cite{Fel20}.
We extend their result and show that even if we additionally parameterize by the skeleton dimension $\kappa$, the \prob{$k$-Center} problem remains W[1]-hard. Moreover, we prove that under the Exponential Time Hypothesis there is no exact algorithm for \prob{$k$-Center} that has runtime $f(k,hd,pw,\kappa) \cdot \vert V \vert^{o(pw + \kappa + \sqrt{k+hd})}$ for any computable function $f$.
\end{abstract}

\section{Introduction}
The \prob{$k$-Center} problem consists of the following task: Given a graph $G=(V,E)$ with positive edge weights $\ell: E \rightarrow \mathbb{Q}^+$ and some $k \in \mathbb{N}$, choose $k$ center vertices $C \subseteq V$ that minimize the maximum distance from any vertex of the graph to the closest center. Formally, if the shortest path distances in $G$ are given by $\dist\colon V^2 \rightarrow \mathbb{Q}^+$ and $B_r(v) = \{w \in V \mid \dist(v,w) \leq r\}$ denotes the ball of radius $r$ around $v$, we aim for a solution $C \subseteq V$ of size $\vert C \vert \leq k$ that has minimum cost, which is the smallest radius $r \geq 0$ such that $V = \bigcup_{v \in C} B_r(v)$.

On general graphs, the \prob{$k$-Center} problem is NP-complete~\cite{Vazirani2001}, as well as on planar graphs~\cite{Plesnik1980} and geometric graphs using $L_1$-, $L_2$- or $L_\infty$-distances~\cite{Feder1988}. On the positive side there is a general $2$-approximation algorithm by Hochbaum and Shmoys~\cite{Hochbaum1986}, i.e.\ an efficient algorithm that computes a solution which deviates from the optimum at most by a factor of $2$. This factor is tight, as for any $\epsilon > 0$, it is NP-hard to compute a $(2-\epsilon)$-approximation, even when considering planar graphs~\cite{Plesnik1980} or graphs with $L_1$- or $L_\infty$-distances~\cite{Feder1988}.

However, common applications of the \prob{$k$-Center} problem arise in domains like logistics or urban planning. For instance, one might want to place a limited number of warehouses, hospitals or police stations on a map such that the distance from any point to the closest facility is minimized. Hence, it is natural to study the problem on transportation networks. Common characterizations of such networks are graphs that are planar or have low doubling dimension,  highway dimension, or skeleton dimension. For formal definitions of these parameters, see \cref{sec:preliminaries}. Usually, it is assumed that in transportation networks the mentioned parameters are bounded by $\mathcal{O}(\mathrm{polylog} \vert V \vert)$ or $\mathcal{O}(\sqrt{\vert V \vert})$. It was shown that on graphs of maximum degree $\Delta$ and highway dimension $hd$, the skeleton dimension is at most $(\Delta + 1) \cdot hd$ ~\cite{kos16}.
The relationship between highway dimension $hd$ and skeleton dimension $\kappa$ was also evaluated experimentally on several real-world road networks and it turned out that $\kappa \ll hd$ \cite{blu18}.
Moreover, it was conjectured that on road networks the skeleton dimension is a constant whereas the highway dimension grows faster than $\mathcal{O}(\mathrm{polylog} \vert V \vert)$.

Still, a low highway dimension or skeleton dimension does not suffice to overcome the general inapproximability bound of \prob{$k$-Center}. In particular, it was shown that for any $\epsilon > 0$, there is no $(2-\epsilon)$-approximation algorithm for graphs of highway dimension $hd \in \mathcal{O}(\log^2 \vert V \vert)$~\cite{DBLP:journals/algorithmica/Feldmann19} or skeleton dimension $\kappa \in \mathcal{O}(\log^2 \vert V \vert)$~\cite{DBLP:conf/iwpec/000119}, unless P=NP.

Apart from approximation, a common way of dealing with NP-hard problems is the use of fixed-parameter algorithms. Such an algorithm computes an exact solution in time $f(p) \cdot n^{\mathcal{O}(1)}$, where $f$ is a computable function and $p$ a parameter of the problem instance which is independent of the problem size $n$. 
In other words, if a problem admits a fixed-parameter algorithm, the complexity of the problem can be captured through some parameter $p$. If this is the case, we call the problem fixed-parameter tractable (FPT).
A natural parameter for \prob{$k$-Center} is the number of center vertices $k$. However, it was shown that in general, \prob{$k$-Center} is W[2]-hard for parameter $k$, and hence it is not fixed-parameter tractable unless W[2] = FPT~\cite{Demaine2005}.
Feldmann and Marx studied the fixed-parameter tractability of \prob{$k$-Center} on transportation networks~\cite{Fel20}. They showed that \prob{$k$-Center} is W[1]-hard even if the input is restricted to planar graphs of constant doubling dimension and the parameter is a combination of $k$, the highway dimension $hd$ and the pathwidth $pw$. Moreover, they proved that under the Exponential Time Hypothesis (ETH) there is no exact algorithm with runtime $f(k, pw, hd) \cdot \vert V \vert^{o(pw + \sqrt{k + hd})}$.
In the present paper we extend their result and show that one can additionally parameterize by the skeleton dimension $\kappa$ without affecting W[1]-hardness. Formally, we show the following theorem.

\begin{theorem}\label{thm:main_result}
On planar graphs of constant doubling dimension, the \prob{$k$-Center} problem is $W[1]$-hard for the combined parameter $(k,pw,hd,\kappa)$ where $pw$ is the pathwidth, $hd$ the highway dimension and $\kappa$ the skeleton dimension of the input graph.
Assuming ETH there is no $f(k,pw,hd,\kappa) \cdot \vert V \vert^{o(pw+\kappa+\sqrt{k+hd})}$ time algorithm\footnote{Here $o(pw+\kappa+\sqrt{k+hd})$ stands for $g(pw+\kappa+\sqrt{k+hd})$ where $g$ is a function with $g(x) \in o(x)$} for any computable function $f$.
\end{theorem}

The reduction of Feldmann and Marx produces a graph where the maximum degree $\Delta$ can be quadratic in the input size. As we have $\Delta \leq \kappa$, it does not imply any hardness for the skeleton dimension. Our new construction yields a graph of constant maximum degree, which enables us to bound the skeleton dimension as well as the highway dimension and the pathwidth.

The results reported by Blum and Storandt~\cite{blu18} indicate that in real-world road networks, the skeleton dimension $\kappa$ is significantly smaller than the highway dimension, which motivates the use of $\kappa$ as a parameter.  
Note that in general, the parameters $pw$, $hd$ and $\kappa$ are incomparable~\cite{DBLP:conf/iwpec/000119}. Still, our main result shows that combining all these parameters and the number of centers $k$ does not allow a fixed-parameter algorithm unless FPT=W[1]. 
However, for the combined parameters $(k,hd)$ \cite{DBLP:journals/algorithmica/Feldmann19} and $(k,\kappa)$ \cite{Fel20}, the existence of a fixed-parameter approximation algorithm was shown, i.e.\ an approximation algorithm with runtime $f(p) \cdot n^{\mathcal{O}(1)}$ for parameter $p$. \Cref{thm:main_result} indicates that apart from approximation there is not much hope for efficient algorithms.

\section{Preliminaries}\label{sec:preliminaries}
For $n \in \mathbb{N}$, let $[n] = \{1, \dots, n\}$. Addition modulo $4$ is denoted by $\modplus$. For $(a,b), (a',b') \in \mathbb{N}$ let $(a,b) \leq (a',b')$ iff $a < a'$ or $a = a'$ and $b \leq b'$.

In a graph $G=(V,E)$ we denote the shortest $s$-$t$ path by $\pi(s,t)$ and the length of a path $P$ by $\vert P \vert$. The concatenation of two paths $P$ and $P'$ is denoted by $P \circ P'$.

A graph $G$ is planar if it can be embedded into the plane without crossing edges, and $d$-doubling if for any $r > 0$, any ball $B_{2r}(v)$ of radius $2r$ in $G$ is contained in the union of $d$ balls of radius $r$. If $d$ is the smallest integer such that $G$ is $d$-doubling, the graph $G$ has doubling dimension $\log_2 d$.

For the highway dimension several slightly different definitions can be found in the literature~\cite{abr10,abr11,abr16}. Here we use the one given in \cite{abr10}.

\begin{definition}
The highway dimension of a graph $G$ is the smallest integer $hd$ such that for any radius $r$ and any vertex $v$ there is a hitting set $S \subseteq B_{4r}(v)$ of size $hd$ for the set of all shortest paths $\pi$ satisfying $\vert \pi \vert > r$ and $\pi \subseteq B_{4r}(v)$.
\end{definition}

To define the skeleton dimension, which was introduced in \cite{kos16},
we need to consider the geometric realization $\tilde G$ of a graph $G$. Intuitively, $\tilde G$ is a continuous version of $G$ where every edge is subdivided into infinitely many infinitely short edges. For a vertex $s \in V$, let $T_s$ be the shortest path tree of $s$. We assume that in $G$ every shortest path is unique, which can be achieved, e.g., by slightly perturbing the edge weights, and it follows that $T_s$ is also unique. The skeleton $T^*_s$ is defined as the subtree of $\tilde T_s$ induced by all $v \in \tilde V$, for which there is some vertex $w$ such that $v$ is contained in $\pi(s,w)$ and moreover, we have $\dist(s,v) \leq 2 \cdot \dist(v,w)$.

\begin{definition}
For a skeleton $T^*_s = (V^*, E^*)$ and a radius $r>0$, let $\mathrm{Cut}_s^r = \{v \in V^* \mid \dist(s,v) = r \}$.
The skeleton dimension of a graph $G$ is $\kappa = \max_{s, r} \vert \mathrm{Cut}_s^r \vert$.
\end{definition}

We conclude this section with a definition of the pathwidth.

\begin{definition}
A path decomposition of a graph $G = (V,E)$ is a sequence $(X_1, \dots, X_r)$ where every $X_i$ (also called \emph{bag}) is a subset of $V$ and the following properties are satisfied:
\begin{romanenumerate}
\item $\bigcup_{i=1}^r X_i = V$,
\item for every edge $\{u,v\} \in E$ there is a bag $X_i$ containing both $u$ and $v$, and
\item for every three indices $i \leq j \leq$ we have $X_i \cap X_k \subseteq X_j$.
\end{romanenumerate}
The width of a path decomposition is the size of the largest bag minus one, i.e.\ $\max_{i=1}^r \left( \vert X_i \vert - 1 \right)$. The pathwidth $pw$ of a graph $G = (V,E)$ is defined as the minimum width of all path decompositions of $G$.
\end{definition}

\section{The Reduction}
Following the idea of Feldmann and Marx~\cite{Fel20}, who showed that on planar graphs of constant doubling dimension, \prob{$k$-Center} is $W[1]$-hard for parameter $(k, pw, hd)$, we present a reduction from the \prob{Grid Tiling with Inequality} ($\prob{GT}_\leq$) problem. This problem asks the following question:
Given $\chi^2$ sets $S_{i,j} \subseteq [n]^2$ of pairs of integers, where $(i,j) \in [\chi]^2$, is it possible to choose one pair $s_{i,j} \in S_{i,j}$ from every set, such that
\begin{itemize}
		\item if $s_{i,j} = (a,b)$ and $s_{i+1,j} = (a',b')$ we have $a \leq a'$, and
		\item if $s_{i,j} = (a,b)$ and $s_{i,j+1} = (a',b')$ we have $b \leq b'$.
\end{itemize}
It is known that the $\prob{GT}_\leq$ problem is $W[1]$-hard for parameter $\chi$ and, unless the Exponential Time Hypothesis (ETH) fails, it has no $f(\chi) \cdot n^{o(\chi)}$ time algorithm for any computable $f$~\cite{Cygan2015}.

\subsection{The Reduction of Feldmann and Marx}
In \cite{Fel20} the following graph $H_{\mathcal I}$ is constructed from an instance $\mathcal{I}$ of $\prob{GT}_\leq$. For any of the $\chi^2$ sets $S_{i,j}$, the graph $H_{\mathcal I}$ contains a gadget $H_{i,j}$ that consists of a cycle $O_{i,j} = v_1 v_2 \dots v_{16n^2+4} v_1$ and five additional vertices $x^1_{i,j},x^2_{i,j},x^3_{i,j},x^4_{i,j}$ and $y_{i,j}$. Every edge contained in some cycle $O_{i,j}$ has unit length and every vertex $y_{i,j}$ is connected to $O_{i,j}$ via edges to $v_1, v_{4n^2+2}, v_{8n^2+3}$ and $v_{12n^2+4}$, which all have length $2n^2+1$. Moreover, for every pair $(a,b) \in S_{i,j}$ and $\tau = (a-1) \cdot n + b$, the gadget $H_{i,j}$ contains the four edges
\begin{itemize}
	\item $\{x^1_{i,j},v_\tau\}$ of length $2n^2-\frac{a}{n+1}$,
	\item $\{x^2_{i,j},v_{\tau+4n^2+1}\}$ of length $2n^2+\frac{b}{n+1}-1$,
	\item $\{x^3_{i,j},v_{\tau+8n^2+2}\}$ of length $2n^2+\frac{a}{n+1}-1$, and
	\item $\{x^4_{i,j},v_{\tau+12n^2+3}\}$ of length $2n^2-\frac{b}{n+1}$.
\end{itemize}

Finally, the individual gadgets are connected in a grid-like fashion, which means that there is a path from $x^2_{i,j}$ to $x^4_{i,j+1}$ and from $x^3_{i,j}$ to $x^1_{i+1,j}$. Each of these paths has length $1$ and consists of $n+2$ edges of length $\frac{1}{n+2}$.

Feldmann and Marx showed that the given $\prob{GT}_\leq$ instance $\mathcal{I}$ has a solution if and only if the \prob{$k$-Center} problem in the graph $H_{\mathcal I}$ has a solution of cost $2n^2$ using $k = 5 \chi^2$ centers. Moreover, the graph $H_{\mathcal I}$ is planar and has doubling dimension $\mathcal{O}(1)$, highway dimension $\mathcal{O}(\chi^2)$ and pathwidth $\mathcal{O}(\chi)$.
Observe that the degree of any vertex $x^h_{i,j}$ is $\vert S_{i,j} \vert$. This means that the skeleton dimension of $H_{\mathcal{I}}$ might be as large as $\Omega(n^2)$, as the maximum degree of $H_{\mathcal I}$ is a lower bound on its skeleton dimension. We show now how to construct a graph $G_\mathcal{I}$ that resembles $H_\mathcal{I}$, but has skeleton dimension $\mathcal{O}(\chi)$ and fulfills the other mentioned properties. 

\subsection{Our Construction}
We assume that in the given $\prob{GT}_\leq$-instance, for all $(i,j) \in [\chi]^2$ and every $b \in [n]$, there is some $a \in [n]$ such that $(a,b) \in S_{i,j}$. This is a valid assumption, as from an instance $\mathcal I$ of ordinary $\prob{GT}_\leq$, we can construct the following instance $\mathcal I'$. For $i \in [\chi-1]$ and $j \in [\chi]$ we add the pairs $\left\{ (n+\chi-i,b) \mid b \in [n] \right\}$ to $S_{i,j}$. Moreover, we add the pairs $\left\{(0,b) \mid b \in [n] \right\}$ to every $S_{\chi,j}$. 
Clearly, every solution for $\mathcal I$ is also a solution for $\mathcal I'$. Consider now a solution for $\mathcal I'$. For $(i,j) \in [\chi-1] \times [\chi]$ we cannot choose a dummy pair $s_{i,j} = (n+\chi-i,b)$, as there is no $(a',b') \in S_{i+1,j}$ such that $a' \geq n+\chi-i$. Moreover, it is not possible to choose $s_{\chi,j} = (0,b')$ as $S_{\chi-1,j}$ contains no pair $(a,b)$ satisfying $a \leq 0$. Hence, $\mathcal I$ has a solution if and only if $\mathcal I'$ has a solution. 

Given a $\prob{GT}_\leq$-instance $\mathcal{I}$ we construct the following graph $G_{\mathcal{I}}$ (cf. \Cref{fig:reduction}). Like in~\cite{Fel20}, we create a gadget $G_{i,j}$ for every set $S_{i,j}$. %
Any $G_{i,j}$ contains a cycle $O_{i,j}$, which initially consists of four edges that have length $2^{n+2} + \nicefrac{1}{n}$. Denote the four vertices of the cycle $O_{i,j}$ by $z^1_{i,j}, \dots, z^4_{i,j}$ and for $h \in [4]$ let $O^h_{i,j} = \shp{z^h_{i,j}, z^{\suc h}_{i,j}}$.  Now, for any pair $(a,b) \in S_{i,j}$ and any $h \in [4]$ we insert a vertex $v^h_{(a,b)}$ into the path $O^h_{i,j}$ and place it such that its distance to $z^h_{i,j}$ is
\[ d_{(a,b)} = 2^{b} - 1 + \frac{a}{n}. \]
It follows that the distance between $v^h_{(a,b)}$ and $v^{\suc h}_{(a,b)}$ is $2^{n+2} + \nicefrac{1}{n}$. Moreover, for $(a',b') \leq (a,b)$, the distance from $v^h_{(a',b')}$ to $v^h_{(a,b)}$ is $2^b - 2^{b'} + \nicefrac{(a-a')}{n}$.

\begin{figure}[t]
\begin{subfigure}[c]{0.5\textwidth}
	\centering
	\includegraphics{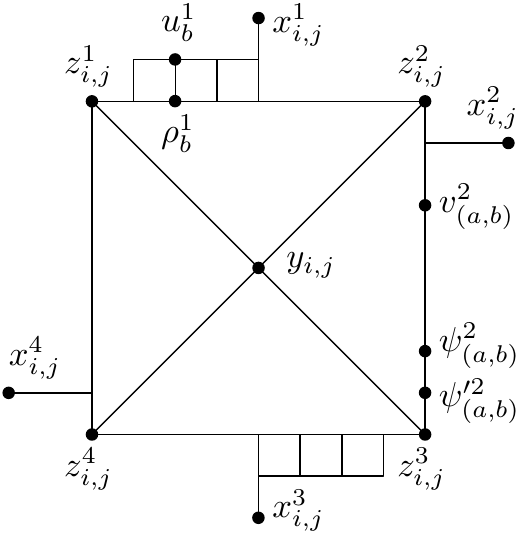}
	\caption{A gadget $G_{i,j}$.}\label{fig:gadget}
\end{subfigure}
\hfill
\begin{subfigure}[c]{0.5\textwidth}
	\centering
	\includegraphics{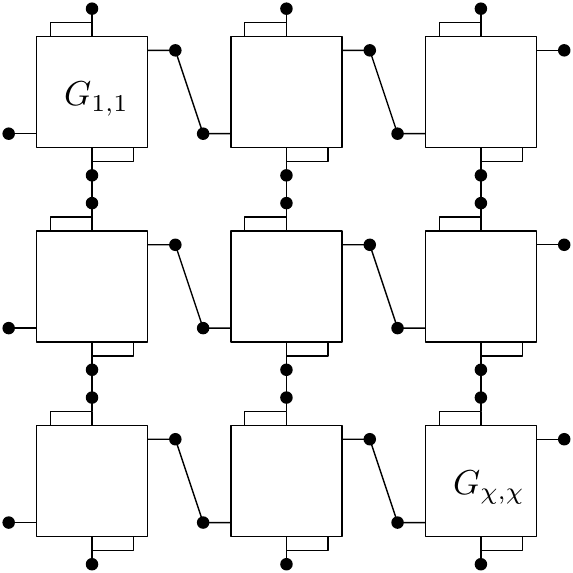}
	\caption{The graph $G_{\mathcal{I}}$.}\label{fig:gadget}
\end{subfigure}
\caption{A single gadget $G_{i,j}$ and the whole graph $G_{\mathcal{I}}$.}
\label{fig:reduction}
\end{figure}

Additionally, for any pair $(a,b) \in S_{i,j}$ and any $h \in [4]$, we insert two vertices $\psi^{h}_{(a,b)}$ and $\psi'^{h}_{(a,b)}$ into the path $O^h_{i,j}$ such that their distance from $v^h_{(a,b)}$ is $2^{n+1}$ and $2^{n+1} + \nicefrac{1}{n}$, respectively. 
This implies that
$\dist\left(\psi'^{h}_{(a,b)}, v^{\suc h}_{(a,b)}\right)= 2^{n+1}$ and $\dist\left(\psi^{h}_{(a,b)}, v^{\suc h}_{(a,b)}\right)= 2^{n+1} + \nicefrac{1}{n}$.

Any gadget $G_{i,j}$ also contains a central vertex $y_{i,j}$ that is connected to each $z^h_{i,j}$ through an edge of length $2^{n+1}+1$. Finally, we add four vertices $x^1_{i,j}, \dots, x^4_{i,j}$ to every gadget $G_{i,j}$, through which we will connect the individual gadgets. 
For $(a,b) \in S_{i,j}$ and $h \in [4]$ denote the distance between $v^h_{(a,b)}$ and $x^h_{i,j}$ by $d^h_{(a,b)}$.
The idea of our reduction is that we attach every $x^h_{i,j}$ to the cycle $O_{i,j}$ such that for every pair $(a,b) \in S_{i,j}$ and $h \in \{1,3\}$, the distance $d^h_{(a,b)}$ reflects the value of $b$, whereas for $h \in \{2,4\}$, the distance $d^h_{(a,b)}$ reflects the value of $a$.

For the latter, we simply add an edge between $x^2_{i,j}$ and the vertex $v^2_{(a^*,b^*)}$ where $(a^*,b^*) = \min S_{i,j}$.\footnote{here the minimum is taken w.r.t\ the lexical order as defined in the preliminaries} The length of this edge is chosen as 
\[
		d^2_{(a^*,b^*)} = 2^n + 1 + d_{(a^*,b^*)} = 2^n + 2^{b^*} + \frac{a^*}{n}.
\]
Similarly we add the edge $\left\{x^4_{i,j}, v^4_{(a^*,b^*)}\right\}$ and set its length to
\[
		d^4_{(a^*,b^*)} = 2^{n+1} - d_{(a^*,b^*)} = 2^{n+1} + 1 - 2^{b^*} - \frac{a^*}{n}.
\]
It follows that for all $(a,b) \in S_{i,j}$ we have
\begin{align*}
		d^2_{(a,b)} & = d^2_{(a^*,b^*)} + d_{(a,b)} - d_{(a^*,b^*)} = 2^n + 2^b + \frac{a}{n} \quad \text{and}\\
		d^4_{(a,b)} & = d^4_{(a^*,b^*)} + d_{(a^*,b^*)} - d_{(a,b)} = 2^{n+1} + 1 - 2^{b} - \frac{a}{n}.
\end{align*}
Attaching $x^1_{i,j}$ and $x^3_{i,j}$ to $G_{i,j}$ is slightly more elaborate. We want to ensure that for any two pairs $(a,b), (a,b') \in S_{i,j}$ that agree on the first component, we have $d^1_{(a,b)} = d^1_{(a,b')}$. For that purpose, we add a path $U^1_{i,j} = u^1_1\,\dots\,u^1_n$ and set the length of every edge $\{u^1_\lambda,u^1_{\lambda+1}\}$ to $2^{\lambda}$. Moreover, we add the edge $\{u^1_n,x^1_{i,j}\}$ of length $2^{n}$. For every $b \in [n]$, consider the vertex $v^1_{(a^*,b)}$ that is furthest from $z^1_{i,j}$. 
We call it also the \emph{$b$-portal} $\rho^1_b$. We attach it to $u^1_b$ through an edge of length $2^{b} - \nicefrac{a^*}{n}$, the so called \emph{$b$-portal edge}. It follows that for $(a,b) \in S_{i,j}$ we have
$
		\dist(v^1_{(a,b)}, u^1_b) = 2^{b} - \nicefrac{a^*}{n} + d_{(a^*,b)} - d_{(a,b)} = 2^{b} - \nicefrac{a}{n}
$
and
$
		\dist(u^1_b, x^1_{i,j}) = \sum_{\lambda = b}^{n} 2^{\lambda} = 2^{n+1} - 2^{b},
$
and hence we have
\[
	d^1_{(a,b)} = 2^{n+1} - \frac{a}{n}.
\]

Similarly we proceed with the vertices contained in $O^3_{i,j}$. We add a path $U^3_{i,j} = u^3_1\,\dots\,u^3_n$, set the length of every edge $\{u^3_\lambda,u^3_{\lambda+1}\}$ to $2^{\lambda}$ and add the edge $\{u^3_n,x^3_{i,j}\}$ of length $2^{n}$. For $b \in [n]$ we use the vertex $v^3_{(a^*,b)}$ that is closest to $z^3_{i,j}$ as the $b$-portal $\rho^3_b$ and attach it to $u^3_b$ trough a portal edge of length $2^{b} - 1 + \nicefrac{a^*}{n}$. It follows that
\[
	d^3_{(a,b)} = 2^{n+1} -1 + \frac{a}{n}.
\]

To complete the construction, we connect the individual gadgets in a grid-like fashion. For $i \in [n-1]$ we connect $x^3_{i,j}$ and $x^1_{i+1,j}$ through a path $P_{i,j}$ of length $1$ that consists of $(n+1)$ edges of length $\nicefrac{1}{(n+1)}$ each. Moreover, for $j \in [n-1]$ we connect $x^2_{i,j}$ and $x^4_{i,j+1}$ through a path $P_{i,j}' = w_1 \dots w_{n}$ where $w_1 = x^4_{i,j+1}$ and $w_{n} = x^4_{i,j}$. We set the length of every edge $\{w_{\lambda+1}, w_{\lambda}\}$ to $2^{\lambda}$ which implies that $\vert P'_{i,j} \vert = 2^n-2$.
The resulting graph $G_{\mathcal I}$ can be constructed in polynomial time from the given $\prob{GT}_\leq$-instance $\mathcal{I}$.

\subsection{Graph Properties}
We now show some basic properties of $G_{\mathcal I}$ that will be useful to prove the correctness of our reduction and to obtain bounds on several graph parameters. 
We first observe that all shortest paths between the cycle $O_{i,j}$ and a path $U^h_{i,j}$ have a certain structure (cf. \cref{fig:sp_structure}).

\begin{lemma}\label{lem:sp-structure}
Let $a,b,b' \in [n]$ and $h \in \{1,3\}$. For $\beta \in [n]$ denote the path $\shp{v_{(a,b)}^h, \rho^h_\beta} \circ \left\{\rho^h_\beta, u^h_\beta\right\} \circ \shp{u^h_\beta, u^h_{b'}}$ by $P_\beta$.
\begin{alphaenumerate}
\item \label{lem:sp-structure:a} If $b' \geq b$, the shortest path from $v_{(a,b)}^h$ to $u_{b'}^h$ is $P_b$.
\item If $b' < b$, the shortest path from $v_{(a,b)}^h$ to $u_{b'}^h$ is $P_{b'}$.
\end{alphaenumerate}
\end{lemma}
\begin{proof}
Any shortest path from $v_{(a,b)}^h$ to $u_{b'}^h$ needs to contain some portal edge $\left\{\rho^h_\beta, u^h_\beta\right\}$. We only prove case~(\ref{lem:sp-structure:a}) for $h=1$, the remaining cases can be shown similarly.

Let $\beta \in [n]$ and let $\rho^1_\beta = v^1_{(\alpha,\beta)}$ be the $\beta$-portal. The path $P_\beta$ has length $\dist(v_{(a,b)}^1, \rho^1_\beta) + \dist(\rho^1_\beta, u^1_\beta) + \dist(u^1_\beta, u^1_{b'}) = \left\vert 2^\beta - 2^b + \nicefrac{(\alpha-a)}{n} \right\vert + 2^\beta - \nicefrac{\alpha}{n} + \left\vert 2^{b'} - 2^\beta \right\vert$.

This means that $\vert P_b \vert = 2^{b'} - \nicefrac{a}{n}$, and for $\beta < b$ we have $\vert P_\beta \vert = 2^{b'} + 2^b - 2^\beta + \nicefrac{(a - 2\alpha)}{n} > 2^{b'} - \nicefrac{a}{n}$. Let now $\beta > b$. If $\beta \leq b'$, we obtain that $\vert P_\beta \vert  = 2^{b'} + 2^{\beta} - 2^b - \nicefrac{\alpha}{n}$ while for $\beta > b'$, we have $\vert P_\beta \vert = 3 \cdot 2^{\beta} - 2^b - 2^{b'} - \nicefrac{a}{n}$, which is both greater than $2^{b'} - \nicefrac{a}{n}$. Hence, $P_b$ is the shortest path from $v_{(a,b)}^h$ to $u_{b'}^h$. 
\end{proof}

\begin{figure}[t]
\begin{subfigure}[c]{0.5\textwidth}
	\centering
	\includegraphics[page=1]{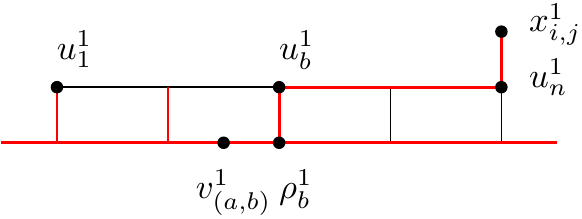}
	\caption{The shortest path tree of a vertex $v^1_{(a,b)}$.}\label{fig:sp_structure_O}
\end{subfigure}
\hfill
\begin{subfigure}[c]{0.5\textwidth}
	\centering
	\includegraphics[page=2]{sp_structure}
	\caption{The shortest path tree of a vertex $u^1_b$.}\label{fig:sp_structure_U}
\end{subfigure}
\caption{Illustration of the shortest path structure as shown in \cref{lem:sp-structure}.}
\label{fig:sp_structure}
\end{figure}
Moreover, it holds that for any vertex $v$ of the graph $G_\mathcal{I}$, there is some central vertex $y_{i,j}$ not too far away. 
\begin{lemma}\label{lem:dist_closest_center}
For every vertex $v \in V$, we have $\min_{(i,j)} \dist(v,y_{i,j}) \leq 2^{n+2} + 2^{n+1}$.
\end{lemma}
\begin{proof}
Assume first that $v$ is contained in some gadget $G_{i,j}$. If $v$ is contained in the cycle $O_{i,j}$, the distance to the closest vertex $z^h_{i,j}$ is at most $2^{n+1}$ as every edge length is a multiple of $\nicefrac{1}{n}$ and the subpath $O^h_{i,j}$ between $z^h_{i,j}$ and $z^{\suc h}_{i,j}$ has length $2^{n+2} + \nicefrac{1}{n}$. Moreover, we have $\dist(z^h_{i,j}, y_{i,j}) = 2^{n+1} + 1$, and hence, the distance between any $v \in O_{i,j}$ and $y_{i,j}$ is bounded by $2^{n+2} +1$.

Consider now some vertex $x^h_{i,j}$. The distance from $x^h_{i,j}$ to any vertex $v^h_{(a,b)}$ is $d^h_{(a,b)}$ and the length of the path from  $d^h_{(a,b)}$ to $y_{i,j}$ via $z^h_{i,j}$ is $d_{(a,b)} + 2^{n+1} +1 $. It follows, that
\begin{align}\label{eq:dist_xy}
	\dist(x^h_{i,j}, y_{i,j}) \leq d^h_{(a,b)} + d_{(a,b)} + 2^{n+1} + 1 \leq 2^{n+2} + 2^{n} + 2,
\end{align}
where the last inequality follows from the fact that $d^h_{(a,b)} + d_{(a,b)} \leq 2^{n+1} + 2^n + 1$, which is easy to verify.
Assume now that $v \in u^h_b$ for some $h \in \{1,3\}$ and $b \in [n]$. The shortest path $\pi(u^h_b, y_{i,j})$ passes through the portal edge $\{u^h_b, \rho^h_b\}$ of length at most $2^b$ and the vertex $z^h_{i,j}$. The distance from $\rho^h_b$ to $z^h_{i,j}$ is at most $2^b$ and it follows that
\[
		\dist(u^h_b, y_{i,j}) \leq 2^b + 2^b + 2^{n+1} + 1 \leq 2^{n+2} + 1.
\]
It remains to consider the case where $v$ is not contained in any gadget. If this holds, $v$ is contained in some path $P_{i,j}$ or $P'_{i,j}$ between two gadgets. The lengths of these paths is bounded by $2^n - 2$ and hence, there is some vertex $x^h_{i,j}$ such that $\dist(v, x^h_{i,j}) \leq 2^n - 2$. It follows from \cref{eq:dist_xy}, that $\dist(v,y_{i,j}) \leq 2^{n+2} + 2^{n+1}$.
\end{proof}

\subsection{Correctness of the Reduction}
We show now that the $\prob{GT}_\leq$-instance $\mathcal{I}$ has a solution if and only if the \prob{$k$-Center} instance $G_{\mathcal I}$ has a solution of cost at most $2^{n+1}$ for $k = 5 \chi^2$ centers.

\begin{lemma}\label{lem:correctness1}
A solution for the $\prob{GT}_\leq$-instance $\mathcal{I}$ implies a solution for the \prob{$k$-Center} instance $G_{\mathcal I}$ of cost at most $2^{n+1}$.
\end{lemma}
\begin{proof}
For $(i,j) \in [n]^2$ let $s_{i,j}$ be the pair from $S_{i,j}$ that is chosen in a solution of $\mathcal{I}$. For the \prob{$k$-Center} instance $G_{\mathcal I}$, we choose a center set $C$ of size $5 \chi^2$ by selecting from every gadget $G_{i,j}$ the central vertex $y_{i,j}$ and the four vertices $v^1_{s_{i,j}}, \dots, v^4_{s_{i,j}}$.
We show that $C$ has cost at most $2^{n+1}$.

Consider a gadget $G_{i,j}$ and the four chosen centers $v^1_{(a,b)}, \dots, v^4_{(a,b)}$. It holds that the distance between any two neighboring centers $v^h_{(a,b)}$ and $v^{\suc h}_{(a,b)}$ is $2^{n+2} + \nicefrac{1}{n}$ and moreover, the length of every edge of the cycle $O_{i,j}$ is a multiple of $\nicefrac{1}{n}$. Hence, it follows that for every vertex $v \in O_{i,j}$ there is some center vertex $v^h_{(a,b)}$ at distance at most $2^{n+1}$. Consider some vertex $u^h_{b'}$ for $h \in \{1,3\}$. It follows from \cref{lem:sp-structure}, that $\dist(v_{(a,b)}^h, u^h_{b'}) \leq 2^{n+1}$. Finally, the vertex $y_{i,j}$ is chosen as a center. This means that the complete gadget $G_{i,j}$ is contained in the five balls of radius $2^{n+1}$ around $y_{i,j}$ and $v^1_{(a,b)}, \dots, v^4_{(a,b)}$.

It remains to show that the chosen centers cover all paths $P_{i,j}$ and $P_{i,j}'$ that connect the individual gadgets. Consider two neighboring gadgets $G_{i,j}$ and $G_{i+1,j}$ and let $s_{i,j} = (a,b)$ and $s_{i+1,j} = (a',b')$ be the corresponding pairs from the solution of $\mathcal{I}$. We have $a \leq a'$. 
From $G_{i,j}$ we have chosen a center $v^3_{(a,b)}$ that has distance $d^3_{(a,b)}$ to $x^3_{i,j}$.  Similarly, we have chosen some $v^1_{(a',b')}$ from $G_{i+1,j}$ whose distance to $x^1_{i+1,j}$ is $d^1_{(a',b')}$. The path $P_{i,j}$ between $x^3_{i,j}$ and $x^1_{i+1,j}$ has length $1$, and hence the distance between the two considered centers is 
\[d^3_{(a,b)} + \vert P_{i,j} \vert + d^1_{(a',b')}= 2^{n+1} - 1 + \frac{a}{n} + 1 + 2^{n+1} - \frac{a'}{n} = 2^{n+2} + \frac{a-a'}{n} \leq 2^{n+2} .\]
This means that $P_{i,j}$ can be covered with balls of radius $2^{n+1}$ around $v^1_{(a',b')}$ and $v^3_{(a,b)}$.
Similarly, $b \leq b'$ yields
\begin{align*}
		d^2_{(a,b)} + \vert P_{i,j}' \vert + d^4_{(a',b')} & = 2^n + 2^b + \frac{a}{n} + 2^n - 2 + 2^{n+1} + 1 - 2^{b'} - \frac{a'}{n} \\
		& = 2^{n+2} - 1+ 2^b - 2^{b'} + \frac{a-a'}{n} < 2^{n+2}
\end{align*}
Hence, any vertex contained in a path $P_{i,j}'$ has distance at most $2^{n+1}$ from a chosen center.
\end{proof}

Moreover, every solution for $G_{\mathcal{I}}$ of cost at most $2^{n+1}$ contains four equidistant vertices $v^1_{(a,b)}, \dots, v^4_{(a,b)}$ from every $G_{i,j}$, which yield a solution for $\mathcal{I}$.
The following lemma completes our correctness proof.

\begin{lemma}\label{lem:correctness2}
A solution for the \prob{$k$-Center} instance $G_{\mathcal{I}}$ of cost at most $2^{n+1}$ implies a solution for the $\prob{GT}_\leq$-instance $\mathcal{I}$.
\end{lemma}
\begin{proof}
Let $C$ be a solution for $G_{\mathcal{I}}$ of cost at most $2^{n+1}$. Consider a gadget $G_{i,j}$. The central vertex $y_{i,j}$ has distance at least $2^{n+1}+1$ to any other vertex. Hence we have $y_{i,j} \in C$.  Let $C_{i,j}$ be the remaining centers from $C$ that have distance at most $2^{n+1}$ from any vertex of $G_{i,j}$. As $k = 5 \chi^2$, there are at most $4 \chi^2$ such centers in total.
We first show that every $C_{i,j}$ consists of exactly $4$ vertices contained in the cycle $O_{i,j}$, such that any two consecutive vertices have distance $2^{n+2} + \nicefrac{1}{n}$.

\begin{claim}
	For $(i,j) \in [n]^2$ we have $C_{i,j} \subseteq O_{i,j}$ and $\vert C_{i,j} \vert = 4$.
\end{claim}
\begin{claimproof}
Let $(A,B) = \max S_{i,j}$ and let $h \in [4]$. 
We show that $\psi^h_{(A,B)}$ can only be covered through vertices from $G_{i,j} \setminus \{x^1_{i,j}, \dots, x^4_{i,j}\}$.

Consider some vertex $x^{h'}_{i,j}$. The shortest path from $x^{h'}_{i,j}$ to $\psi^h_{(A,B)}$ has to pass trough either $v^h_{(A,B)}$ or $z^h_{i,j}$. The distance from $\psi^h_{(A,B)}$ to $v^h_{(A,B)}$ is $2^{n+1}$ whereas the distance from $\psi^h_{(A,B)}$ to $z^h_{i,j}$ is $2^{n+1} + \nicefrac{1}{n} - 2^B + 1 - \nicefrac{A}{n} > 2^n$. Moreover, the distance from $x^{h'}_{i,j}$ to any vertex in the cycle $O_{i,j}$ is at least $2^n$. It follows that $\dist( x^{h'}_{i,j}, \psi^h_{(A,B)}) > 2^{n+1}$ and hence, $\psi^h_{(A,B)}$ cannot be covered through $x^{h'}_{i,j}$ or any vertex not contained in the gadget $G_{i,j}$.

Moreover, any two of the vertices $\psi^1_{(A,B)}, \dots, \psi^4_{(A,B)}$ have distance at least $2^{n+2} + \nicefrac{1}{n}$ and hence we need at least $4$ centers to cover them with balls of radius $2^{n+1}$. This implies that $C_{i,j} \subseteq G_{i,j} \setminus \{x^1_{i,j}, \dots, x^4_{i,j}\}$ and $\vert C_{i,j} \vert = 4$.

Assume now that $C_{i,j} \not \subseteq O_{i,j}$, which means that some vertex $u^{h}_{b} \in C_{i,j}$ was chosen as a center.  Let $v^h_{(a,b)}$ be the corresponding $b$-portal. \Cref{lem:sp-structure} implies that the distance from $u^{h}_{b}$ to any of the vertices $\psi^1_{(a,b)}, \dots, \psi^4_{(a,b)}$ is more than $2^{n+1}$. Moreover, the pairwise distance of $\psi^1_{(a,b)}, \dots, \psi^4_{(a,b)}$ is at least $2^{n+2} + \nicefrac{1}{n}$. This means that apart from $u^{h}_{b}$, the set $C_{i,j}$ needs to contain $4$ more centers, which contradicts $\vert C_{i,j} \vert = 4$. Hence we obtain $C_{i,j} \subseteq O_{i,j}$.
\end{claimproof}
We now show, that every $C_{i,j}$ contains four equidistant centers $v^h_{(a,b)}$.
\begin{claim}
	For $(i,j) \in [n]^2$ we have $C_{i,j} = \left\{ v^1_{(a,b)}, \dots, v^4_{(a,b)} \right\}$ for $(a,b) \in S_{i,j}$.
\end{claim}

\begin{claimproof}
Let $(\alpha,\beta)$ be the minimum of $S_{i,j}$. Consider the vertex $x^1_{i,j}$. Its distance to $z^1_{i,j}$, $\psi^1_{(\alpha,\beta)}$ and any vertex of $O_{i-1,j}$ is more than $2^{n+1}$. Hence, it must be covered through some vertex $v^1_{(a,b)}$ where $(a,b) \in S_{i,j}$. %
Consider the vertices $\psi'^1_{(a,b)}, \psi^2_{(a,b)}, \psi'^2_{(a,b)}, \psi^3_{(a,b)}, \psi'^3_{(a,b)}, \psi^4_{(a,b)}$. None of them is contained in the ball of radius $2^{n+1}$ around $v^1_{(a,b)}$.
Moreover, for $h \in \{1,2,3\}$, the distance between $\psi'^h_{(a,b)}$ and $\psi^{\suc h}_{(a,b)}$ is $2^{n+2}$, whereas the distance between $\psi'^1_{(a,b)}$ and $\psi'^2_{(a,b)}$ and the distance between $\psi^3_{(a,b)}$ and $\psi^4_{(a,b)}$ are both $2^{n+2} + \nicefrac{1}{n}$. This means that complete cycle $O_{i,j}$ can only be covered with $4$ balls of radius $2^{n+1}$ if we have $\left\{v^2_{(a,b)}, v^3_{(a,b)}, v^4_{(a,b)}\right\} \subseteq C_{i,j}$.
\end{claimproof}
Finally we show that the sets $C_{i,j}$ yield a solution for the $\prob{GT}_\leq$-instance $\mathcal{I}$.
\begin{claim}
For $(i,j) \in [n]^2$ choosing $s_{i,j} = (a,b)$ where $v^1_{(a,b)} \in C_{i,j}$ yields a solution for the $\prob{GT}_\leq$-instance $\mathcal{I}$.
\end{claim}
\begin{claimproof}
		Let $s_{i,j} = (a,b)$ and $s_{i+1,j} = (a',b')$ and assume that $a > a'$. Consider the path $P_{i,j}$ connecting the vertices $x^3_{i,j}$ and $x^1_{i+1,j}$. As the path $P_{i,j}$ consists of $n+1$ edges of length $\nicefrac{1}{(n+1)}$, it contains a vertex $w$ that has distance $1 - \nicefrac{a}{(n+1)}$ from $x^3_{i,j}$ and distance $\nicefrac{a}{(n+1)}$ from $x^1_{i+1,j}$. It follows that the distances from $w$ to the closest centers in $G_{i,j}$ and $G_{i+1,j}$ are
\begin{align*}
		& d^3_{(a,b)} + 1 - \frac{a}{n+1} = 2^{n+1} - 1 + \frac{a}{n} + 1 - \frac{a}{n+1} > 2^{n+1} \quad \text{and} \\
		& d^1_{(a',b')} + \frac{a}{n+1} = 2^{n+1} - \frac{a'}{n} + \frac{a}{n+1} > 2^{n+1},
\end{align*}
respectively. This contradicts the fact that $C$ is a solution for the \prob{$k$-Center} instance, and hence $a \leq a'$.
Similarly, let $s_{i,j} = (a,b)$ and $s_{i,j+1} = (a',b')$ and assume that $b > b'$. Consider the path $P'_{i,j} = w_1 \dots w_n$ connecting $x^2_{i,j}$ and $x^4_{i,j+1}$. Recall that every edge $\{w_{\lambda+1}, w_{\lambda}\}$ has length $2^{\lambda}$ and hence we have $\dist(x^2_{i,j}, w_b) = 2^{n} - 2^b$ and $\dist(w_b, x^4_{i,j+1}) = 2^b - 2$.
It follows that the distances from $w_{b}$ to the closest centers in $G_{i,j}$ and $G_{i,j+1}$ are
\begin{align*}
		d^2_{(a,b)} + 2^n - 2^b & = 2^n + 2^b + \frac{a}{n} + 2^n - 2^b = 2^{n+1} + \frac{a}{n} > 2^{n+1} \quad \text{and}\\
		d^4_{(a',b')} + 2^b - 2 & = 2^{n+1} + 1 - 2^{b'} + \frac{a'}{n} + 2^b - 2 \geq 2^{n+1} + \frac{a'}{n} > 2^{n+1},
\end{align*}
respectively, which gives a contradiction. It follows that $b \leq b'$ and hence, choosing $s_{i,j} = (a,b)$ for $v^1_{(a,b)} \in C_{i,j}$ yields a solution for $\mathcal{I}$.
\end{claimproof}
This completes the proof as any solution $C$ of cost at most $2^{n+1}$ for the \prob{$k$-Center} instance $G_{\mathcal{I}}$ implies a solution for the $\prob{GT}_\leq$-instance $\mathcal{I}$.
\end{proof}

\section{Bounds on Graph Parameters}
In this section we show bounds on the doubling dimension, the highway dimension, the skeleton dimension and the pathwidth of the graph $G_{\mathcal{I}}$, which imply \cref{thm:main_result}. 
To bound the doubling dimension, we exploit the fact that the individual gadgets $G_{i,j}$ are connected in a grid-like fashion. This means that we can bound the diameter of balls within this grid.
For that purpose, let $A_{i,j}(d) = \{ (i',j') \in [\chi]^2 \mid |i'-i| + |j'-j| \leq d \}$. Moreover, let $V_{i,j}(d)$ be the vertices of all gadgets $G_{i',j'}$ satisfying $(i',j') \in A_{i,j}(d)$ and the vertices on the paths $P_{i',j'}$ and $P'_{i',j'}$ between these gadgets.
We now bound the diameter of the graph induced by $V_{i,j}(d)$.

\begin{lemma}\label{lem:diam_Vij}
Consider the graph induced by $V_{i,j}(d)$. Its diameter is at most $(2^{n+3} + 2^{n+1} + 2^n + 2) \cdot (2d + 1)$. Moreover, if $\vert A_{i,j}(d) \vert = (2d+1)^2$, i.e.\ $A_{i,j}(d)$ contains all possible index pairs, the diameter is at least $(2^{n+2} + 2^n) \cdot (2d+1)$.
\end{lemma}

\begin{proof}
Let $(i,j) \in [\chi]$. We first bound the distance between any $x^h_{i,j}$ and $x^{h'}_{i,j}$.

\begin{claim}
For $h, h' \in [4]$ where $h \neq h'$ we have $2^{n+2} + 2^n < \dist(x^h_{i,j}, x^{h'}_{i,j}) \leq 2^{n+3} + 2^{n+1} + 4$. 
\end{claim}
\begin{claimproof}
The upper bound follows directly from \cref{eq:dist_xy} in the proof of \cref{lem:dist_closest_center}. For the lower bound, observe that the shortest path between $x^h_{i,j}$ and $x^{h'}_{i,j}$ needs to pass through two vertices $v^h_{(a,b)}$ and $v^{h'}_{(a',b')}$ of the cycle $O_{i,j}$. It holds that the distance from $x^h_{i,j}$ to $v^h_{(a,b)}$ and from $x^{h'}_{i,j}$ to $v^{h'}_{(a',b')}$ are $d^h_{(a,b)}$ and $d^{h'}_{(a',b')}$, which are both at least $2^n$.
Moreover, the distance from $v^h_{(a,b)}$ to $v^{h'}_{(a',b')}$ is minimized, if $v^h_{(a,b)} = v^h_{(n,n)}$ and $v^{h'}_{(a',b')} = v^{\suc h}_{(1,1)}$. As $\dist(v^h_{(n,n)}, v^{\suc h}_{(1,1)}) = 2^{n+2} + \nicefrac{1}{n} + d_{(1,1)} - d_{(n,n)} = 2^{n+1} + 2^{n} + 1 + \nicefrac{1}{n}$, a lower bound of $2 \cdot 2^n + 2^{n+1} + 2^{n} + 1 + \nicefrac{1}{n} > 2^{n+2} + 2^n$ on $\dist(x^h_{i,j}, x^{h'}_{i,j})$ follows.
\end{claimproof}
Consider the graph induced by $V_{i,j}(d)$. Any shortest path in this graph traverses at most $2d+1$ gadgets and contains at most $2d$ paths between two gadgets. These paths have length at most $2^n-2$. Moreover, it follows from the proof of \cref{lem:dist_closest_center} that the diameter of a single gadget is at most $2^{n+3}+2^{n+1}+4$.
This means that the distance of any shortest path is upper bounded by $(2d+1) \cdot (2^{n+3}+2^{n+1}+2^n+2)$.

If $\vert A_{i,j}(d) \vert = (2d+1)^2$, the shortest path from $x^1_{i-d,j}$ to $x^4_{i+d,j}$ has to traverse $2d+1$ gadgets hence a lower bound of $(2d+1) \cdot (2^{n+2}+2^n)$ on the diameter of the graph induced by $V_{i,j}(d)$ follows.
\end{proof}

This allows us to show that the doubling dimension of $G_{\mathcal{I}}$ is constant.
\begin{lemma}
	The graph $G_{\mathcal{I}}$ is planar and has constant doubling dimension.
\end{lemma}
\begin{proof}
It can be seen easily that $G_{\mathcal{I}}$ is planar. Recall that a graph has doubling dimension at most $d$ if any ball of radius $2r$ can be covered with $2^d$ balls of radius $r$.

To bound the doubling dimension of $G_{\mathcal{I}}$, consider a ball $B_{2r}(v)$ of radius $2r$ around some vertex $v \in V$. \Cref{lem:dist_closest_center} implies that there is a vertex $y_{i,j}$ satisfying $\dist(v,y_{i,j}) \leq 2^{n+2} + 2^{n+1}$.  It follows that the ball $B_{2r}(v)$ is contained in the ball around $y_{i,j}$ that has radius $2^{n+2} + 2^{n+1} + 2r$. Moreover, \cref{lem:diam_Vij} implies that the latter ball in turn is contained in $V_{i,j}(d)$ if $2 \cdot (2^{n+2} + 2^{n+1} + 2r) \leq (2^{n+2} + 2^n) \cdot (2d+1)$. This is true for $2d+1 = 6r / 2^{n+2}$ and $r \geq 2^{n+2} + 2^{n+1}$.

We now show that we can cover the vertices $V_{i,j}(d)$ through a constant number of balls that have radius $r$ and are centered at vertices $y_{i',j'}$. \Cref{lem:diam_Vij} implies that for every $(i',j') \in [\chi]^2$, the ball $B_r(y_{i',j'})$ contains the set $V_{i',j'}(d')$ if $2r \geq (2^{n+3} + 2^{n+1} + 2^n + 2) \cdot (2d' + 1)$. This is the case for $2d'+1 = 2r / 2^{n+4}$. As we want $V_{i',j'}(d')$ to be nonempty, we require $d' \geq 0$, which holds for $r \geq 2^{n+3}$. Hence it suffices to show that there is a constant number of sets  $V_{i',j'}(d')$ whose union contains $V_{i,j}(d)$.

As it was observed in~\cite{Fel20}, the index set $A_{i,j}(d)$ is contained in the union of $\lceil \frac{2d+1}{2d'+1} \rceil^2$ index sets $A_{i',j'}(d')$. It follows that we can cover the vertices $V_{i,j}(d)$ through $\lceil \frac{2d+1}{2d'+1} \rceil^2$ vertex sets $V_{i',j'}(d')$. 
Hence, for $r \geq 2^{n+3}$ we can cover $B_{2r}(v)$ with $\lceil \frac{2d+1}{2d'+1} \rceil^2 = \lceil \frac{6r/2^{n+2}}{2r/2^{n+4}} \rceil^2 = 144$ balls of radius $r$.

Assume now that $r < 2^{n+3}$. We already showed that $B_{2r}(v)$ is contained in $V_{i,j}(d)$ if $2d +1  = 6r/2^{n+2} < 12$, which implies $d < 6$. Hence, the ball $B_{2r}(v)$ intersects at most $\vert A_{i,j}(5) \vert \leq (2 \cdot 5 + 1)^2 = 121$ gadgets $G_{i',j'}$. 
We show that we can cover any of these gadgets $G_{i',j'}$ and the paths to its neighboring gadgets through a constant number of ball $B_r(w)$.

If $r \geq 2^{n+1}$, we can choose the $9$ balls centered at $y_{i',j'}, z^h_{i',j'}$ and $ x^h_{i',j'}$ where $h \in [4]$, as for every $w \in O_{i',j'}$ there is some $z^h_{i',j'}$ satisfying $\dist(z^h_{i',j'}, w) \leq 2^{n+1}$, for $h \in \{1,3\}$ and $b \in [n]$ it holds that $\dist(x^h_{i',j'},u^h_b) \leq 2^{n+1}$ and the length the paths to the neighboring gadgets have length at most $2^n-2$.

Let now $r < 2^{n+1}$. If $v = y_{i',j}$, i.e.\ the ball $B_{2r}(v)$ is centered at $y_{i',j}$,  we can choose $X_{i',j'} = \{y_{i',j'}, z^h_{i',j'} \mid h \in [4]\}$. Otherwise, $B_{2r}(v) \cap O_{i',j'}$ is a subpath of $O_{i',j'}$ that has length at most $4r$, which can be covered by $4$ balls of radius $r$. Similarly, we can also cover $B_{2r}(v) \cap U^h_{i',j'}$, $B_{2r}(v) \cap P_{i',j'}$ and $B_{2r}(v) \cap P'_{i',j'}$ with $4$ balls of radius $r$ each. This means that we can cover $B_{2r} \cap G_{i',j'}$ through a constant number of balls of radius $r$.

It follows that we can cover any ball $B_{2r}(v)$ for any $v \in V$ and any $r > 0$ with a constant number of balls of radius $r$, which completes the proof.
\end{proof}

We next bound the highway dimension of $G_{\mathcal{I}}$.

\begin{lemma}\label{lem:bound_highway_dimension}
The graph $G_{\mathcal{I}}$ has highway dimension $hd \in \mathcal{O}(\chi^2)$.
\end{lemma}

\begin{proof}
		For any radius $r>0$ we specify a set $H_r$ such that every shortest path $\pi$ satisfying $\vert \pi \vert > r$ intersects $H_r$ and moreover, for every vertex $v \in V$ we have $\vert H_r \cap B_{4r}(v) \vert \in \mathcal{O}(\chi^2)$. Let $X = \{ y_{i,j}, x^h_{i,j}, z^h_{i,j} \mid (i,j) \in [\chi]^2, h \in [4] \}$. For $r \geq 2^{n+2}$ we choose $H_r = X$. We have $\vert H_r \vert = 9 \chi^2$ and hence for every vertex $v \in V$ we have $\vert H_r \cap B_{4r}(v) \vert \in \mathcal{O}(\chi^2)$.  We show now that any shortest path of length more than $r$ intersects $H_r$. Clearly, all shortest paths that are not completely contained within one single gadget are hit by $H_r$ as all $x^h_{i,j}$ are contained in $H_r$ and the paths $P_{i,j}$ and $P'_{i,j}$ between the individual gadgets have length at most $2^n - 2$. Consider some gadget $G_{i,j}$. All edges of the cycle $O_{i,j}$ have length at least $\nicefrac{1}{n}$ and for any $h \in [4]$ we have $\dist(z^h_{i,j}, z^{\suc h}_{i,j}) = 2^{n+2} + \nicefrac{1}{n}$. Hence, any subpath of $O_{i,j}$ that has length at least $2^{n+2}$ intersects $H_r$. Moreover, for $h \in \{1,3\}$, the path $U^h_{i,j}$ has length $2^n-2$. 
	
It remains to consider some shortest path $\pi(s,t)$ where $s \in O_{i,j}$ and $t \in U^h_{i,j}$. Let $t = u^h_b$.
According to \cref{lem:sp-structure}, the shortest path $\pi(s,t)$ traverses exactly one portal edge $\{\rho^h_\beta, u^h_\beta\}$ where $\beta \in [b]$. 
This means that $\dist(s,t) = \dist(s, \rho^h_\beta) + \dist(\rho^h_\beta, u^h_b) \leq \dist(s, \rho^h_\beta) + 2^b$.
The vertex $s$ is contained in the shortest path $\pi(z^h_{i,j}, \rho^h_\beta)$ or in $\pi(\rho^h_\beta, z^{\suc h}_{i,j})$. 
In the first case we have $\dist(s,\rho^h_\beta) < \dist(z^h_{i,j}, \rho^h_\beta) \leq 2^\beta$. This implies that $\dist(s,t) < 2^\beta + 2^b \leq 2^{n+1}$.
In the second case we have $\dist(s,\rho^h_\beta) \leq 2^{n+2} - 2^\beta$ and moreover \cref{lem:sp-structure} implies that $\beta = b$. Hence we obtain $\dist(s,t) \leq 2^{n+2} - 2^\beta + 2^\beta = 2^{n+2}$.
This means that every shortest path of length more than $r \geq 2^{n+2}$ is hit by $H_r$.

Let now $r < 2^{n+2}$. For a shortest path $p = v_1 \dots v_\nu$ and $q >0$ let $p^{\langle q \rangle}$ be a \emph{$q$-cover} of $p$, i.e.\ we have $p^{\langle q \rangle} \subseteq \{v_1, \dots, v_\nu\}$ such that any subpath of $p$ that has length at least $q$ contains some node from $p^{\langle q \rangle}$. We consider $q$-covers $p^{\langle q \rangle}$ that are constructed greedily, i.e.\ we start with $p^{\langle q \rangle} = \{v_1\}$ and iteratively add the closest vertex that has distance at least $q$. For $(i,j) \in [\chi]^2$ let 
\[ X_{i,j} = \bigcup_{h \in [4]} {O^h_{i,j}}^{\langle r/4 \rangle} \cup \bigcup_{h \in \{1,3\}} {U^h_{i,j}}^{\langle r/4 \rangle} \cup \left\{ u^1_n, u^3_n \right\} \cup P_{i,j}^{\langle r/4 \rangle} \cup {P'_{i,j}}^{\langle r/4 \rangle} \]
and choose $H_r = X \cup \bigcup_{(i,j) \in [\chi]^2} X_{i,j}$.
Consider some shortest path $\pi(s,t)$ that has length more than $r$. Clearly, $\pi(s,t)$ is hit by $H_r$ if it contains some node from $X$ or it is a subpath of some cycle $O_{i,j}$, some path $U^h_{i,j}$ or some path $P_{i,j}$ or $P'_{i,j}$. It remains to be shown that $\pi(s,t)$ is also hit by $H_r$ if $s \in O_{i,j}$ and $t \in U^h_{i,j}$. Let $t = u^h_b$. \Cref{lem:sp-structure} implies that $\pi(s,t)$ consists of a subpath $p$ of $O_{i,j}$, a portal edge $\{\rho^h_\beta, u^h_{\beta}\}$ and a subpath $p'$ of $U^h_{i,j}$. Assume that $\pi(s,t)$ is not hit by $H_r$. By the choice of $X_{i,j}$ we have $\vert p \vert < \nicefrac{r}{4}$ and $\vert p' \vert < \nicefrac{r}{4}$. This means that $\dist(\rho^h_\beta, u^h_\beta) > \nicefrac{r}{2}$. By construction of the graph $G_{\mathcal{I}}$ we have $\dist(\rho^h_\beta, u^h_\beta) \leq 2^\beta$ and hence $2^\beta > \nicefrac{r}{2}$. As we have $u^h_{\beta} \not \in X_{i,j}$, it holds that $\beta \not \in  \{1, n\}$ and moreover it follows from the choice of ${U^h_{i,j}}^{\langle r/4 \rangle}$, that $\dist(u^h_{\beta-1}, u^h_{\beta}) \leq \nicefrac{r}{4}$. However, by construction of $G_{\mathcal{I}}$ we have $\dist(u^h_{\beta-1}, u^h_{\beta}) = 2^{\beta-1}$, which implies $2^\beta \leq \nicefrac{r}{2}$, a contradiction to $2^\beta > \nicefrac{r}{2}$. This means that every shortest path of length more than $r$ is hit by $H_r$.

Finally we have to show that for every vertex $v \in V$ we have $\vert H_r \cap B_{4r}(v) \vert \in \mathcal{O}(\chi^2)$. As for the $\nicefrac{r}{4}$-cover of some shortest path $p$ we have $\vert B_{4r}(v) \cap p^{\langle r/4 \rangle} \vert \in \mathcal{O}(1)$, it follows that for every $(i,j) \in [\chi]^2$ we have $\vert B_{4r}(v) \cap X_{i,j} \vert \in \mathcal{O}(1)$. Moreover there are $\chi^2$ different sets $X_{i,j}$ and we have $\vert X \vert = 9 \chi^2$, which implies $\vert H_r \cap B_{4r}(v) \vert \in \mathcal{O}(\chi^2)$.
\end{proof}

Observe, that for any graph $G$ of highway dimension $hd$ and maximum degree $\Delta$, an upper bound of $(\Delta +1)hd$ on the skeleton dimension of $G$ follows~\cite{kos16}.
As the graph $G_{\mathcal{I}}$ has maximum degree $\Delta = 4$, it follows that the skeleton dimension of $G_{\mathcal{I}}$ is bounded by $\mathcal{O}(\chi^2)$.

However, with some more effort, we can show a stronger bound of $\mathcal{O}(\chi)$.
We will use the following lemma, which was shown in~\cite{blu18}.
\begin{lemma}\label{lem:skel_desc}
Consider vertices $u,v,w \in V$ such that $v$ is contained in $\shp{u,w}$. If $w$ is contained in the skeleton of $u$, it is also contained in the skeleton of $v$.
\end{lemma}

We now bound the size of a skeleton within a single gadget.
For simplicity, in the following we confuse a graph $G$ and its geometric realization $\tilde G$.

\begin{lemma}\label{lem:skel_gadget}
For any $(i,j) \in [\chi]^2$ and any vertex $s$ contained in $G_{i,j}$, the subtree of the skeleton $T_s^*$ induced by the vertices of $G_{i,j}$ is the union of a constant number of paths.
\end{lemma}

\begin{proof}
We first show that every skeleton contains only a limited number of portal edges.
Recall that the skeleton of a shortest path tree is defined on the geometric realization, where every edge is subdivided into infinitely many infinitely short edges. We refer to vertices that were introduced during this subdivision as \emph{interior vertices}.

\begin{claim}\label{claim:skel_cycle}
Consider a vertex $s = v^h_{(a,b)}$ for $(a,b) \in [n]^2$ and $h \in \{1,3\}$. If the skeleton $T^*_s$ of $s$ contains an interior vertex of a portal edge $\{\rho^h_\beta, u^h_\beta\}$, we have $\beta \in \{b,b-1\}$.
\end{claim}
\begin{claimproof}
Assume $h = 1$. For $\beta > b$, it follows from \cref{lem:sp-structure} that $\{\rho^h_\beta, u^h_\beta\}$ is not contained in the shortest path tree of $s$ and hence, no interior vertex of $\{\rho^h_\beta, u^h_\beta\}$ can be contained in $T^*_s$. Let $\beta < b-1$ and let $\rho^h_\beta = v^h_{(\alpha, \beta)}$. \Cref{lem:sp-structure} implies that $u^h_\beta$ is the furthest descendant of $\rho^h_\beta$ in the shortest path tree $T_s$, and we have $\dist(\rho^h_\beta, u^h_b) < 2^{\beta}$. Moreover, the distance from $v^h_{(a,b)}$ to $\rho^h_\beta$ is $d_{(a,b)} - d_{(\alpha, \beta)} = 2^b - 2^\beta + \nicefrac{(a-\alpha)}{n} > 2^{\beta+1} > \nicefrac{1}{2} \cdot \dist(\rho^h_\beta, u^h_b)$. This means that no interior vertex of $\{\rho^h_\beta, u^h_\beta\}$ can be contained in $T^*_s$. The case $h=3$ can be shown similarly.
\end{claimproof}

\begin{claim}\label{claim:skel_comb}
Consider a vertex $s = u^h_b$ for $b \in [n]$ and $h \in \{1,3\}$. If the skeleton $T^*_s$ of $s$ contains an interior vertex of a portal edge $\{\rho^h_\beta, u^h_\beta\}$, we have $\beta \in \{b,b-1,1\}$.
\end{claim}
\begin{claimproof}
Assume $h = 1$. For $\beta > b$, it follows from \cref{lem:sp-structure} that $\{\rho^h_\beta, u^h_\beta\}$ is not contained in the shortest path tree of $s$ and hence, no interior vertex of $\{\rho^h_\beta, u^h_\beta\}$ can be contained in $T^*_s$. Let $1 < \beta < b-1$ and let $\rho^h_\beta = v^h_{(\alpha, \beta)}$. It follows from \cref{lem:sp-structure} that the furthest possible descendant of $u^h_\beta$ within the shortest path tree of $s$ is $v^h_{(1,\beta)}$, which has distance $2^\beta - \nicefrac{1}{n}$ from $u^h_\beta$. The distance from $u^h_b$ to $u^h_\beta$ is $\dist(u^h_b, u^h_\beta) = 2^b - 2^\beta > 2^{\beta+1} > \dist(u^h_\beta,v^h_{(1,\beta)})$ and hence, no interior vertex of $\{\rho^h_\beta, u^h_\beta\}$ can be contained in $T^*_s$. The case $h=3$ can be shown similarly.
\end{claimproof}

Now, let $T_s$ be the shortest path tree of $s$ and $T_s[G_{i,j}]$ be the subtree of $T_s$ induced by the vertices of $G_{i,j}$.
If we disregard all portal edges $\{\rho^h_b, u^h_b\}$, it follows from \cref{lem:sp-structure} that $T_s[G_{i,j}]$ consists of a constant number of subpaths of the cycle $O_{i,j}$, of the two paths $U^1_{i,j}$ and $U^3_{i,j}$, the $4$ edges incident to $x^1_{i,j}, \dots, x^4_{i,j}$ and some of the edges incident to $y_{i,j}$. Hence, if we do not count the portal edges, the subtree $T_s^*[G_{i,j}]$ of the skeleton $T_s^*$ induced by $G_{i,j}$ consists of a constant number of paths. It remains to be shown that $T_s^*[G_{i,j}]$ intersects only a constant number of portal edges.

Assume that $s = \rho^h_b$ for $h \in \{1,3\}$ and $b \in [n]$. It follows from \cref{claim:skel_cycle} that only $\{\rho^h_{b-1}, u^h_{b-1}\}$ and $\{\rho^h_{b}, u^h_{b}\}$ can intersect $T_s^*[G_{i,j}]$. Consider now a portal edge $\{\rho^{h \modplus 2}_\beta, u^{h \modplus 2}_{\beta}\}$ on the opposite side of the cycle. It follows from \cref{lem:skel_desc} that it can only intersect $T_s^*[G_{i,j}]$, if it also intersects the skeleton of $\rho^{h \modplus 2}_1$ or $\rho^{h \modplus 2}_n$, which according to \cref{claim:skel_comb} holds only for $\beta \in \{1, n-1, n\}$. This means that $T_s^*[G_{i,j}]$ consists of a constant number of paths. If $s \in U^h$ for $h \in \{1,3\}$, we can show the same using \cref{claim:skel_comb} and \cref{lem:skel_desc}.

Assume now that $s$ is contained in the cycle $O_{i,j}$, but not a portal. Let $u$ and $v$ be the two closest portals such that $s$ is contained in $\pi(u,v)$. It follows from \cref{lem:skel_desc}, that $T_s^*[G_{i,j}]$ is a subgraph of $T_u^*[G_{i,j}] \cup T_v^*[G_{i,j}] \cup \pi(u,v)$ and hence, it is the union of a constant number of paths. For similar reasons, the same holds if $s = y_{i,j}$ or $s = x^h_{i,j}$.
\end{proof}

Moreover we can show that every cut in any skeleton of $G_{\mathcal{I}}$ intersects at most $\mathcal{O}(\chi)$ different gadgets and connecting paths between two gadgets.

\begin{lemma}\label{lem:gadgets_in_cut}
For every vertex $s \in V$ and every radius $r > 0$, $\mathrm{Cut}_s^r$ intersects $\mathcal{O}(\chi)$ gadgets $G_{i,j}$ and $\mathcal{O}(\chi)$ paths $P_{i,j}$ and $P'_{i,j}$.
\end{lemma}

\begin{proof}
It can be shown that for any $(i,j) \in [\chi]^2$, we have $\dist(y_{i,j}, y_{i+1,j}) = 2^{n+3} + 4 + \nicefrac{2}{n}$ and $\dist(y_{i,j}, y_{i,j+1}) = 2^{n+3} + 3 + \nicefrac{2}{n}$. This means that for any $(i,j), (i',j') \in [\chi]^2$ we have 
\begin{align}\label{eq:dist_central}
		\dist(y_{i,j}, y_{i',j'}) = \vert i - i' \vert \cdot (2^{n+3} + 4 + \nicefrac{2}{n}) + \vert j -j' \vert \cdot (2^{n+3} + 3 + \nicefrac{2}{n}).
\end{align}

Let $r > 0, s \in V$ and consider a vertex $v \in \mathrm{Cut}_s^r$. It holds that $\dist(s,v) = r$. According to \cref{lem:dist_closest_center} there are two central vertices $y_{i,j}$ and $y_{i',j'}$ satisfying $\dist(s, y_{i,j}) \leq 2^{n+2} + 2^{n+1}$ and $\dist(v, y_{i',j'}) \leq 2^{n+2} + 2^{n+1}$.
Using the triangle inequality we obtain that $\dist(y_{i,j}, y_{i',j'}) \in [r^-, r^+]$ where $r^- = r - (2^{n+3} + 2^{n+2})$ and $r^+ = r + 2^{n+3} + 2^{n+2}$.
Moreover, the ball around $y_{i',j'}$ of radius $2^{n+2}+2^{n+1}$ intersects $\mathcal{O}(1)$ gadgets $G_{i'',j''}$ and $\mathcal{O}(1)$ paths $P_{i'',j''}$ and $P'_{i'',j''}$. This means that any bound on the size of the set $Y = \{y_{i',j'} \mid \dist(y_{i,j}, y_{i',j'}) \in [r^-, r^+] \}$ yields a bound on the number of gadgets and paths intersecting $\mathrm{Cut}_s^r$.

Consider now a vertex $y_{i',j'} \in Y$. Assume that $i' \geq i$ and consider some $i^* \geq i'+4$. It follows from \Cref{eq:dist_central} and $\dist(y_{i,j}, y_{i',j'}) \geq r^-$ that
\begin{align*}
	\dist(y_{i,j}, y_{i^*,j'}) \geq \dist(y_{i,j}, y_{i',j'}) + 4 \cdot (2^{n+3} + 4 + \nicefrac{2}{n}) > r^+.
\end{align*}
This means that $y_{i^*,j'} \not \in Y$ and it follows that for any $j' \in [\chi]$ we have $\vert \{ i^* \geq i \mid y_{i^*,j'} \in Y \} \vert \leq 3$. Similarly we can show that $\vert \{ i^* \leq i \mid y_{i^*,j'} \in Y \} \vert \leq 3$ for any $j' \in [\chi]$ . This implies $\vert Y \vert \in \mathcal{O}(\chi)$, which completes the  proof.
\end{proof}

Combining \cref{lem:skel_gadget,lem:skel_desc,lem:gadgets_in_cut}, we obtain that the skeleton dimension of $G_{\mathcal{I}}$ is bounded by $\mathcal{O}(\chi)$.

\begin{lemma}\label{lem:bound_skeleton_dimension}
The graph $G_{\mathcal{I}}$ has skeleton dimension $\kappa \in \mathcal{O}(\chi)$.
\end{lemma}

\begin{proof}
Let $s \in V, r > 0$ and consider $\mathrm{Cut}^r_s$. Any vertex $v \in \mathrm{Cut}^r_s$ is either contained in some gadget $G_{i,j}$ or some connecting path $P_{i,j}$ or $P'_{i,j}$.

We start with bounding the number of vertices that are contained in $\mathrm{Cut}^r_s$ and some $P_{i,j}$ or $P'_{i,j}$. For any $(i,j) \in [\chi]^2$ we have $\vert \mathrm{Cut}^r_s \cap P_{i,j} \vert \leq 2$ as $P_{i,j}$ contains at most two distinct vertices that have the same distance from $s$. For the same reason we have $\vert \mathrm{Cut}^r_s \cap P'_{i,j} \vert \leq 2$. 
Hence \cref{lem:gadgets_in_cut} implies that the size of $\mathrm{Cut}^r_s \cap \{P_{i,j}, P'_{i,j} \mid (i,j) \in [\chi]^2\}$ is bounded by $\mathcal{O}(\chi)$.

Consider now some gadget $G_{i,j}$. We show that $\vert \mathrm{Cut}^r_s \cap G_{i,j} \vert \in \mathcal{O}(1)$. If $s$ is contained in $G_{i,j}$ this follows immediately from \cref{lem:skel_gadget}, as $\mathrm{Cut}^r_s$ intersects any path in $T^*_s$ at most twice. If $s$ is not contained in $G_{i,j}$, \cref{lem:skel_desc} implies that $\mathrm{Cut}^r_s \cap G_{i,j}$ is a subset of 
\begin{align*}
		\left\{ \mathrm{Cut}^{r(h)}_{x^h_{i,j}} \mid h \in [4] \text{ and } r(h) = r - \dist(s,x^h_{i,j}) \right\} \cap G_{i,j}.
\end{align*}
Observe that every $x^h_{i,j}$ is contained in $G_{i,j}$, which means that $\vert \mathrm{Cut}^{r(h)}_{x^h_{i,j}} \cap G_{i,j} \vert \in \mathcal{O}(1)$.
This means that the size of $\mathrm{Cut}^r_s \cap \{ G_{i,j} \mid (i,j) \in [\chi]^2 \}$ is bounded by $\mathcal{O}(\chi)$. Hence we have $\vert \mathrm{Cut}^r_s \vert \in \mathcal{O}(\chi)$ and it follows that $G_{\mathcal{I}}$ has skeleton dimension $\kappa \in \mathcal{O}(\chi)$.
\end{proof}

Finally we bound the pathwidth of the graph $G_{\mathcal{I}}$.

\begin{lemma}\label{lem:bound_treewidth}
The graph $G_{\mathcal{I}}$ has pathwidth $pw \in \mathcal{O}(\chi)$.
\end{lemma}

\begin{proof}

Consider the graph $\hat G_{\mathcal{I}}$ that arises when we contract all vertices of degree $2$ except the vertices $x^h_{i,j}$. It suffices to show that $\hat G_{\mathcal{I}}$ has pathwidth at most $\mathcal{O}(\chi)$.  
For $(i,j) \in [\chi]^2$ denote the gadget $G_{i,j}$ and the cycle $O_{i,j}$ after the contraction by $\hat G_{i,j}$ and $\hat O_{i,j}$, respectively.
We first construct a path decomposition of constant width for every $\hat G_{i,j}$. To this end, consider the cycle $\hat O_{i,j}$, which (as every cycle) has a path decomposition where every bag has size at most $3$. For $h \in \{1,3\}$ and $b \in [n]$, add $u^h_b$ to every bag containing the portal $\rho^h_b$. Finally, add $y_{i,j}$ and $x^1_{i,j}, \dots, x^4_{i,j}$ to every bag. This yields a path decomposition of $\hat G_{i,j}$ which has constant width.

We now combine the path decompositions of the gadgets $\hat G_{i,j}$ to a path decomposition of $\hat G_{\mathcal{I}}$. For $(i,j) \in [\chi]^2$, consider the path decomposition of $\hat G_{i,j}$ and add the vertices $\{ x^1_{i',j'}, \dots x^4_{i',j'} \mid 1 \leq (i'-i) \cdot \chi + (j'-j) \leq \chi \}$ to every bag. According to \cref{fig:reduction}, these are the vertices $x^h_{i',j'}$ of the $\chi$ gadgets after $\hat G_{i,j}$ when considering the gadgets row-wise from left to right. Denote the resulting path decomposition by $\mathcal P_{(i-1) \cdot \chi + j}$. We can observe, that its width is bounded $\mathcal{O}(1) + 4 \chi$. Concatenating all these path decompositions as $\mathcal P_1 \mathcal P_2 \dots \mathcal P_{\chi^2}$ then yields a path decomposition of $\hat G_{\mathcal{I}}$ of width $\mathcal{O}(1) + 4 \chi$, which concludes the proof.
\end{proof}

\section{Conclusion}
The properties shown in the previous section now imply \cref{thm:main_result}. As the $\prob{GT}_\leq$ problem is $W[1]$-hard for parameter $\chi$ and we have $k = 5 \chi^2, hd \in \mathcal{O}(\chi^2), \kappa \in \mathcal{O}(\chi)$ and $pw \in \mathcal{O}(\chi)$, it follows that on planar graphs of constant doubling dimension, \prob{$k$-Center} is $W[1]$-hard for parameter $(k,pw,hd,\kappa)$. Assuming ETH there is no $f(\chi) \cdot n^{o(\chi)}$ time algorithm for $\prob{GT}_\leq$ and hence, \prob{$k$-Center} has no $f(k,hd,pw,\kappa) \cdot \vert V \vert^{o(pw + \kappa + \sqrt{k+h})}$ time algorithm unless ETH fails.

It follows that on planar graphs of constant doubling dimension, \prob{$k$-Center} has no fixed-parameter algorithm for parameter $(k,pw,hd,\kappa)$ unless FPT=W[1].
Moreover, it was shown that \prob{$k$-Center} has no efficient $(2-\epsilon)$-approximation algorithm for graphs of highway dimension $hd \in \mathcal{O}(\log^2 \vert V \vert)$~\cite{DBLP:journals/algorithmica/Feldmann19} or skeleton dimension $\kappa \in \mathcal{O}(\log^2 \vert V \vert)$~\cite{DBLP:conf/iwpec/000119}.

Still, combining the paradigms of approximation and fixed-parameter algorithms allows one to compute a $(2-\epsilon)$ approximation for \prob{$k$-Center} on transportation networks. For instance, there is a $3/2$-approximation algorithm that has runtime $2^{\mathcal{O}(k \cdot hd \log hd)} \cdot n^{\mathcal{O}(1)}$ for highway dimension $hd$~\cite{DBLP:journals/algorithmica/Feldmann19} and a $(1 + \epsilon)$-approximation algorithm with runtime $(k^k/\epsilon^{\mathcal{O}(k \cdot d)}) \cdot n^{\mathcal{O}(1)}$ for doubling dimension $d$~\cite{Fel20}. As the doubling dimension is bounded by $\mathcal{O}(\kappa)$, the latter result implies a $(1 + \epsilon)$-approximation algorithm that has runtime $(k^k/\epsilon^{\mathcal{O}(k \cdot \kappa)}) \cdot n^{\mathcal{O}(1)}$.

On the negative side, there is no $(2 - \epsilon)$-approximation algorithm with runtime $f(k) \cdot n^{\mathcal{O}(1)}$ for any $\epsilon > 0$ and computable $f$ unless W[2]=FPT~\cite{Fel20}. It remains open, to what extent the previously mentioned algorithms can be improved.

\bibliographystyle{plain}
\bibliography{bibliography}

\end{document}